\documentclass[lettersize,journal]{IEEEtran}
\usepackage[T1]{fontenc}
\usepackage{textcomp}
\usepackage{balance}
\usepackage{stfloats}
\usepackage{amsmath,amssymb,amsthm,bm}
\usepackage{graphicx,booktabs,array,cite,url}
\usepackage[caption=false,font=normalsize,labelfont=sf,textfont=sf]{subfig}
\usepackage[hidelinks]{hyperref}
\newtheorem{theorem}{Theorem}
\newtheorem{lemma}{Lemma}
\newtheorem{proposition}{Proposition}
\newtheorem{corollary}{Corollary}
\newtheorem{remark}{Remark}
\newcommand{\CN}{\mathcal{CN}}
\newcommand{\E}{\mathbb{E}}
\newcommand{\Prob}{\mathbb{P}}
\newcommand{\R}{\mathbf{R}}
\newcommand{\U}{\mathbf{U}}
\newcommand{\I}{\mathbf{I}}
\newcommand{\h}{\mathbf{h}}
\newcommand{\z}{\mathbf{z}}
\newcommand{\tr}{\operatorname{tr}}

\graphicspath{{figures/}}
\begin{document}
\title{Certified Outage Dimensioning of Fluid Antenna Receivers: Required SNR, Aperture, and Port Count}
\author{Yuanjun Shen,~\IEEEmembership{Member,~IEEE}%
\thanks{Yuanjun Shen is with the National Key Laboratory of Radar Detection and Sensing, Department of Electronic Engineering, Xidian University, Xi'an 710071, China (e-mail: yuanjun.shen@xidian.edu.cn).}%
\thanks{Manuscript received April 19, 2021; revised August 16, 2021.}}

\markboth{Journal of \LaTeX\ Class Files,~Vol.~14, No.~8, August~2021}%
{Shell \MakeLowercase{\textit{et al.}}: A Sample Article Using IEEEtran.cls for IEEE Journals}

\IEEEpubid{0000--0000/00\$00.00~\copyright~2021 IEEE}

\maketitle
\begin{abstract}
We study the SNR, aperture, and port count needed to meet a target outage in a single-user downlink with a fixed transmit antenna and a single-RF fluid antenna receiver. For known spatially correlated Rayleigh channels and ideal port selection, a sampling-free residual certificate controls reduced-rank outage error over a continuous SNR interval. Positive Bessel features support large arrays, and bounded importance sampling converts the certificate into simultaneous design intervals. At 1 bit/s/Hz and outage $10^{-5}$, a four-wavelength, ten-port receiver requires 4.40--4.55 dB. Reducing its port count from 256 to ten costs 0.70--1.25 dB under a joint 95\% probability statement. At 3.5 wavelengths, ten ports require 4.90--5.05 dB. Integer-port screening across 3--8 dB shows how aperture and port requirements change with available SNR; an effective-SNR transformation extends these decisions to other rates. The reported intervals include model and sampling uncertainty under checked Gaussian covariance inputs. A constructive slow-growth rank bound and a 16--1024-port study support compact evaluation as port count grows.

\end{abstract}


\section{Introduction}
\IEEEPARstart{W}{e} consider a single-user downlink in which a base station transmits from a fixed antenna and a user terminal receives through a fluid antenna system (FAS) with one active radio-frequency (RF) chain. In each fading block, the receiver ideally selects the strongest of $N$ spatially correlated ports distributed over an aperture of $W$ wavelengths. This established setting isolates the diversity available from spatial selection \cite{wong2021,khammassi,new2024}. The joint port-fading law is essential because a reliability failure occurs when all candidate ports are weak. Position-adaptive antennas form part of a broader reconfigurable-wireless research direction \cite{mimoevolution2022,wong2022topics,martinez2022liquid}. Hardware progress includes fabricated surface-wave-enabled fluid antennas with measured operation in the 24--30 GHz band \cite{surfacewave2026}, as well as pixel-based antennas that realize rapid switching among effective port states \cite{pixel2025}. These implementations motivate receiver-level reliability analysis under an explicitly specified port-fading model.

Receiver dimensioning requires the SNR needed to meet a target outage $p$ for a specified $(W,N)$, or the aperture and port count needed at an available SNR. Approximation errors matter particularly at targets such as $10^{-5}$. We formulate these requirements as certified SNR intervals and an SNR--aperture--port design map. Nondominated feasible candidates describe resource tradeoffs within that map.

Existing FAS analyses characterize the joint port-fading distribution in several ways. Full-correlation and eigenmode-based analyses clarified the effect of finite aperture and diminishing gains from adding ports \cite{khammassi,new2024}. Spatial block-correlation models simplified outage evaluation \cite{block2024}. Copula approaches separated marginal fading distributions from dependence modeling \cite{copula2024}. Diversity and coded-modulation studies connected ideal selection gains with time mismatch and practical coding \cite{psomas2024}. More recently, geometric analyses and equivalent-mode approximations revisited the connection between aperture and effective dimension \cite{geometric2026,edof2026}. Outage estimates also inform resource-allocation and fairness designs \cite{rsma2026}. Spatial selection also underlies multiuser fluid-antenna multiple access \cite{fama2022}; the present work isolates the single-user reliability calculation.

\IEEEpubidadjcol
Tractable outage models offer useful tools for receiver design, but approximation uncertainty can affect whether a reliability target is met. For $N=12$, $W=0.5$, and $\theta=1$ at 10~dB, the fixed block-model calibration of Section~\ref{sec:block-model-comparison} predicts about $5.2\times10^{-4}$ against a full-covariance reference estimate of $1.8\times10^{-3}$, placing the two on opposite sides of a $10^{-3}$ target, while a rank-three energy truncation of the same covariance differs from that reference by about 5\% (Section~\ref{sec:energy-accuracy}). A channel-energy criterion and a relative-outage criterion thus measure different aspects of approximation accuracy, and neither reports the error that decides feasibility. We therefore complement point estimates with a probability bracket whose upper bound certifies sufficiency and whose lower bound excludes inadequate resources.

The conservative direction of unnormalized Karhunen--Lo\`eve (KL) outage truncation follows from Anderson's inequality and was used in \cite{kl2026}. We build on that upper bound by controlling the residual approximation error. Port-count selection \cite{new2024} and channel reconstruction from pilots or partial CSI \cite{oversampling2025,partial2026} address complementary deployment and acquisition questions. Here, the covariance is specified, selection is ideal, and the computational rank is chosen for a fixed physical array.

\subsection{Question and Contribution}
We quantify how outage uncertainty affects receiver dimensioning through four results.
\begin{enumerate}
\item We obtain necessary and sufficient SNRs for a target outage (P1) and certified feasible, infeasible, or unresolved aperture--port candidates (P2). These decisions form a finite design map; its resource frontier is identified when all relevant candidates are resolved.
\item A multiplicative residual certificate supplies the model-error guarantee without probability integration. A witness at the smallest threshold covers an entire SNR interval for the same covariance; bounded importance sampling supplies separate numerical intervals.
\item A positive Bessel-feature construction supplies controlled residual variances for large linear FAS arrays. An explicit sufficient condition gives a constructive $O(\log\log N/\log\log\log N)$ rank bound at fixed aperture, threshold, and tolerance. This controls the latent dimension of large-$N$ dimensioning calculations; each retained sample still checks all physical ports at cost $O(Nr)$.
\item Fixed experiments connect aperture and port choices to their power cost. Integer-port screening at ten apertures and multiple SNRs, required-SNR intervals for the selected low-port receivers, and an outage-matched diversity index complement the fixed-density aperture curve and the 16--1024-port rank study. The port study reports both certified minima and integer intervals where a decision remains unresolved.
\end{enumerate}
The residual criterion controls relative outage error for the physical-port event, complementing the power-loss criterion in \cite[Prop. 1 and Remark 3]{kl2026}. Gaussian comparison supplies the model bounds, and bounded importance sampling supplies their numerical intervals. Figure~\ref{fig:framework} summarizes the calculation.

\subsection{Organization and Notation}
Section II specifies the channel and design objective. Section III develops the residual certificate, its large-port implementation, and design inversion. Section IV supplies numerical intervals, and Section V presents the design results and controls. Sections VI and VII discuss scope and conclude. The appendices collect auxiliary probability identities and the limits of rank and geometry approximations. Bold capitals and lowercase letters denote matrices and vectors. The notation $\mathbf A\preceq\mathbf B$ means $\mathbf B-\mathbf A$ is positive semidefinite (PSD). The distribution $\CN(\bm0,\R)$ is proper complex Gaussian; $\|\cdot\|_\infty$ is the largest component magnitude. All logarithms are natural unless specified otherwise; $\mathrm{i}=\sqrt{-1}$.

\section{System Model and Approximation Objective}
Consider a single-user narrowband downlink from one fixed base-station antenna to a fluid antenna user terminal with $N$ candidate receive ports and one active RF chain. The channel is constant within a fading block. Its normalized port vector is
\begin{equation}
 \h\sim\CN(\bm0,\R),\qquad \R\succeq\bm0.
 \label{eq:model}
\end{equation}
For the nominal model $R_{ii}=1$. With perfect instantaneous port selection, average SNR per port $\bar\gamma$, and decoding threshold $\theta=2^{R_0}-1>0$ for rate $R_0$ bit/s/Hz, outage is
\begin{equation}
 F_\R(x)=\Prob\{\|\h\|_\infty^2\le x\},\qquad x=\theta/\bar\gamma.
 \label{eq:outage}
\end{equation}
The equality boundary has probability zero for the nondegenerate cases used here. A linear isotropic-scattering example has normalized aperture $W$ and
\begin{equation}
 R_{ij}=J_0\!\left(2\pi W\frac{|i-j|}{N-1}\right),\quad N>1,
 \label{eq:jakes}
\end{equation}
where $J_0$ is the ordinary Bessel function. The bounds below apply to any covariance satisfying \eqref{eq:model}, including anisotropic and planar arrays.

The physical aperture is $W\lambda$, where $\lambda=c/f_c$ is the
free-space wavelength, and the port spacing is $W\lambda/(N-1)$.
For a 28-GHz size example within the operating band of
\cite{surfacewave2026}, $\lambda\approx1.071$ cm: $W=4$ and $8$
correspond to 4.28 and 8.57 cm. At 3.5 GHz the same normalized
apertures are 34.3 and 68.5 cm, respectively. Thus the intended
device size and carrier determine the admissible $W$ range. This
unit conversion leaves the normalized fading model unchanged;
an actual deployment requires its own propagation and antenna response.

Write $\R=\U\bm\Lambda\U^H$, with $\lambda_1\ge\cdots\ge\lambda_N\ge0$. Define
\begin{equation}
 \R_r=\sum_{j=1}^r\lambda_j\mathbf u_j\mathbf u_j^H,
 \quad \mathbf E_r=\R-\R_r,
 \quad e_r=\frac{\sum_{j>r}\lambda_j}{\tr\R}.
 \label{eq:truncation}
\end{equation}
An energy rule chooses the smallest $r$ satisfying $e_r\le\epsilon_E$. This is an average channel-energy criterion. Our criterion instead controls
\begin{equation}
 \frac{U_r(x)-F_\R(x)}{F_\R(x)},\qquad U_r(x)=F_{\R_r}(x).
 \label{eq:target}
\end{equation}
The numerator is nonnegative by the known Gaussian comparison result used below. The diagonal of $\R_r$ is generally below one. Rescaling it to unit diagonal changes the covariance order and invalidates this particular guarantee. The modes in $\U^H\h$ are independent; the physical-port variables in $\h$ generally are not.

\subsection{Receiver-Dimensioning Problems}
The reference rate in the numerical study is $R_0=1$ bit/s/Hz,
which gives $\theta=1$. For any other rate, define
$\gamma_{\rm eff,dB}=\gamma_{0,\rm dB}-10\log_{10}(2^{R_0}-1)$.
Then $x=10^{-\gamma_{\rm eff,dB}/10}$, so a P2 decision at that
effective SNR applies exactly to the specified rate and available
SNR. Changing rate at fixed $\gamma_0$ therefore changes feasibility.

For a fixed geometry family, write $\R(W,N)$ for its specified covariance. We study two design problems under \eqref{eq:outage}:

\noindent\textit{P1: Required SNR.} Given $(W,N)$, $p$, and $\theta$, bracket
\begin{equation}
 \gamma_{\rm req}(W,N)=\inf\{\gamma>0:F_{\R(W,N)}(\theta/\gamma)\le p\}.
 \label{eq:p1}
\end{equation}
\noindent\textit{P2: Aperture and ports.} Given available SNR $\gamma_0$, $p$, $\theta$, and a finite candidate set $\mathcal D$, identify which $(W,N)\in\mathcal D$ meet the target and determine the feasible candidates not dominated in both resources. Alternatively, select a feasible candidate minimizing a declared deployment cost.

A candidate dominates another if its aperture and port count are both no larger and at least one is smaller. The frontier is defined within $\mathcal D$, with unresolved candidates retained when assessing an optimum. The computational rank $r$ and physical port count $N$ are separate design quantities.
\begin{figure*}[t]
\centering\includegraphics[width=\textwidth]{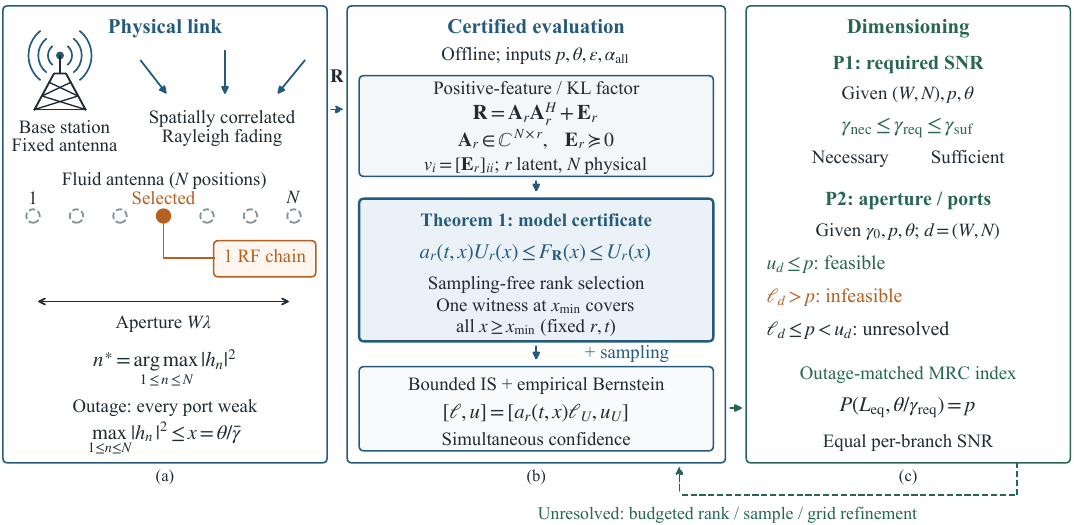}
\caption{Physical link and offline receiver dimensioning. (a) The marked positions are alternatives; the maximizing port $n^{\ast}$ connects to a single RF chain. (b) At fixed covariance, rank, and witness, certification at $x_{\min}$ controls model error for all $x\ge x_{\min}$, equivalently $\bar\gamma\le\theta/x_{\min}$. Here $v_i=[\mathbf E_r]_{ii}$, $r$ is the latent dimension, and $N$ is the physical candidate count. A sampling interval $[\ell_U,u_U]$ for $U_r$ propagates to $[\ell,u]$ for full outage. (c) The resulting intervals support SNR inversion, candidate classification, and a distributional MRC comparison. Dashed arrows carry model information or budgeted refinement.}
\label{fig:framework}
\end{figure*}

\section{Residual Certificates for Receiver Dimensioning}\label{sec:certificate}
\subsection{PSD Comparison and Its Tail Limitation}
Set $\tau_r=\lambda_{r+1}$ for $r<N$ and $\tau_N=0$. The
spectral residual obeys $\bm0\preceq\mathbf E_r\preceq\tau_r\I$, so
\begin{equation}
 \mathbf B_r=\R_r+\tau_r\I,\qquad
 \R_r\preceq\R\preceq\mathbf B_r.
 \label{eq:order}
\end{equation}
\begin{lemma}[Outage bracket]\label{thm:bracket}
For any covariance in \eqref{eq:model}, any $x>0$, and $0\le r\le N$,
\begin{equation}
 L_r(x):=F_{\mathbf B_r}(x)\ \le\ F_\R(x)\ \le\ U_r(x).
 \label{eq:bracket}
\end{equation}
Both bounds coincide with $F_\R$ at $r=N$.
\end{lemma}
\begin{proof}
The outage set $\mathcal D_x=\{\h:|h_i|\le\sqrt{x},\ i=1,\ldots,N\}$
is centrally symmetric and convex in $2N$ real dimensions.
Anderson's inequality \cite{anderson} says that adding independent
centered Gaussian noise cannot increase its probability. Applying
this to \eqref{eq:order} gives both bounds. Singular covariances
follow by adding $\delta\I$ and taking $\delta\downarrow0$.
\end{proof}
The upper bound is the established unnormalized KL conservatism
\cite{kl2026}. The covariance majorant also gives a lower bound,
but may lose the small residual's directional structure. For
$\R\succ0$, the small-ball identity in Appendix~\ref{app:polydisc}
gives $F_\R(x)\sim x^N/\det\R$. Hence
for any $\mathbf M\succeq\R\succ0$,
\begin{equation}
 \lim_{x\downarrow0}\frac{F_\R(x)}{F_{\mathbf M}(x)}
 =\frac{\det\mathbf M}{\det\R}.
 \label{eq:detbarrier}
\end{equation}
For $\mathbf M=\mathbf B_r$, this ratio is
$\prod_{j\le r}(1+\tau_r/\lambda_j)\prod_{j>r}(\tau_r/\lambda_j)$.
If it exceeds $1+\epsilon$, this particular lower model cannot support
that relative-error budget for all sufficiently small thresholds. 

Conditional evaluation and ratio-based rank tests for this
majorant are given in Appendix~\ref{app:majorants}--\ref{app:covrank}.
The certificate below instead uses the individual residual variances
to control relative outage error without evaluating $L_r$.

\subsection{A Sampling-Free Multiplicative Certificate}
Write $\mathbf A_r=[\sqrt{\lambda_1}\mathbf u_1,\ldots,
\sqrt{\lambda_r}\mathbf u_r]$ and $\h=\mathbf A_r\z+\mathbf e$,
where $\z\sim\CN(\bm0,\I_r)$ and the residual
$\mathbf e\sim\CN(\bm0,\mathbf E_r)$ are independent. The reduced
outage probability is
\begin{equation}
 U_r(x)=\E_\z\mathbf1\{\|\mathbf A_r\z\|_\infty^2\le x\}.
 \label{eq:upper_integral}
\end{equation}
Let $v_i=[\mathbf E_r]_{ii}$. For $0<t<1$, define
\begin{align}
 q_r(t,x)&=\sum_{i:v_i>0}\exp(-t^2x/v_i),\\
 a_r(t,x)&=(1-t)^{2r}[1-q_r(t,x)]_+.
 \label{eq:residualcoeff}
\end{align}
\begin{theorem}[Relative certificate from residual variances]\label{prop:residual}
For $1\le r\le N$ and the Gaussian KL decomposition above,
\begin{equation}
 a_r(t,x)U_r(x)\le F_\R(x)\le U_r(x).
 \label{eq:multbracket}
\end{equation}
Thus any witness $t$ satisfying $a_r(t,x)\ge(1+\epsilon)^{-1}$ certifies $(U_r-F_\R)/F_\R\le\epsilon$. If the residual is zero, use $t=0$ and $a_r=1$.
\end{theorem}
\begin{proof}
Write $\h=\mathbf A_r\z+\mathbf e$ with $\z\sim\CN(\bm0,\I_r)$ independent of $\mathbf e$. The events $\|\mathbf A_r\z\|_\infty\le(1-t)\sqrt{x}$ and $\|\mathbf e\|_\infty\le t\sqrt{x}$ imply outage. Independence and the marginal Gaussian tail union bound give
\begin{equation}
 F_\R(x)\ge U_r((1-t)^2x)[1-q_r(t,x)]_+.
\end{equation}
For $s=1-t$ and $D=\{\z:\|\mathbf A_r\z\|_\infty^2\le x\}$, a change of variables in $r$ complex coordinates gives
\begin{align}
 U_r(s^2x)&=s^{2r}\int_D\pi^{-r}e^{-s^2\|\z\|^2}\,d\z\nonumber\\
 &\ge s^{2r}\int_D\pi^{-r}e^{-\|\z\|^2}\,d\z=s^{2r}U_r(x).
\end{align}
Combining the two inequalities proves the lower bound; the upper bound is the known Anderson direction. No independence across residual ports is required.
\end{proof}

The rank search evaluates only the positive sums $v_i=\sum_{j>r}\lambda_j|u_{ij}|^2$ and a scalar witness. A globally optimal $t$ is unnecessary: every feasible witness gives a valid sufficient rank. We use a fixed logarithmic/linear grid with local refinement and evaluate the chosen coefficient at 80 decimal digits. Numerical-input controls are detailed in Section~\ref{sec:protocol}.

For fixed $r,t$, $a_r(t,x)$ is nondecreasing in $x$. A witness passing at $x_{\min}$ therefore guarantees the same model-error budget for all $x\ge x_{\min}$, or all average SNRs up to $\theta/x_{\min}$, at fixed covariance. This deterministic guarantee holds over the continuous SNR range. The union bound and Gaussian dilation can increase the sufficient rank; evaluation cost also depends on the estimator.

\subsection{Large-Port Evaluation with Positive Features}
The residual argument applies more generally whenever
$\R=\mathbf A\mathbf A^H+\mathbf E$ with $\mathbf E\succeq0$ and
$\mathbf A$ having $r$ columns. An independent proper Gaussian
decomposition exists, and Theorem~\ref{prop:residual} uses
$U_{\mathbf A}=F_{\mathbf A\mathbf A^H}$ and $v_i=E_{ii}$.
Consequently, an explicitly controlled positive feature approximation
can be used without computing the numerically tiny eigenvalues of a
dense covariance matrix.

For centered line positions $p_i\in[-W/2,W/2]$, Neumann's addition
formula \cite[Sec. 10.23(ii)]{dlmf} gives
\begin{equation}
 R_{ij}=J_0(2\pi p_i)J_0(2\pi p_j)
       +2\sum_{k=1}^{\infty}J_k(2\pi p_i)J_k(2\pi p_j).
 \label{eq:besselfeature}
\end{equation}
Let $\mathbf G_m$ contain the first $m$ features, including order zero.
Its omitted covariance is PSD. The real-argument bound
$|J_k(z)|\le (|z|/2)^k/k!$, obtained from Poisson's integral
\cite[Eq. 10.9.4]{dlmf}, gives the uniform diagonal bound
\begin{equation}
 d_m(W)=\frac{2b^{2m}}{(m!)^2[1-b^2/(m+1)^2]},\quad b=\pi W/2,
 \label{eq:featuretail}
\end{equation}
provided $m+1>b$. Indeed, the omitted squared-feature series is bounded by a geometric
tail whose successive ratios are at most $b^2/(m+1)^2$.
With the thin SVD $\mathbf G_m=\mathbf V\bm\Sigma\mathbf Q^H$,
retain $\mathbf A_r=\mathbf V_r\bm\Sigma_r$. Then
\begin{equation}
 v_i\le\sum_{j>r}\sigma_j^2|V_{ij}|^2+d_m(W).
 \label{eq:featurevariance}
\end{equation}
Both components are nonnegative; subtracting a nearly unit retained
energy from $R_{ii}=1$ is unnecessary. These are KL modes of the finite
feature covariance.
The omitted-feature term connects their certificate to the original
Jakes covariance.

An explicit, conservative sufficient test also follows. Set
\begin{equation}
 t_\epsilon=1-(1+\epsilon)^{-1/(4r)},\qquad
 \eta_\epsilon=1-(1+\epsilon)^{-1/2}.
\end{equation}
For any residual satisfying $v_i\le\bar v$, the condition
\begin{equation}
 \bar v\le\frac{t_\epsilon^2x}{\log(N/\eta_\epsilon)}
 \label{eq:explicitvariance}
\end{equation}
implies $q_r\le\eta_\epsilon$ and $a_r\ge(1+\epsilon)^{-1}$.
For direct retention of $r$ Bessel features, one may use
$\bar v=d_r(W)$. For the exact rank-$r$ KL truncation, the optimal
trace residual is at most $Nd_r(W)$, so $\bar v=Nd_r(W)$ is a valid,
usually much looser, choice. Numerical residuals in
\eqref{eq:featurevariance} are used for the experiments.

\begin{remark}[Port count and aperture]
Rapid feature-tail decay explains the weak port-count dependence on
the tested grids. The sufficient condition \eqref{eq:explicitvariance}
contains $N$ through the union bound. For example, if a fixed
rank leaves $v_i\ge v_0>0$ at at least $cN$ ports, then for every $t<1$,
$q_r(t,x)\ge cN\exp(-x/v_0)$. This particular certificate eventually
becomes vacuous at fixed $r,x$. This limitation concerns the marginal union bound. Control of the
continuous residual process is needed for an $N$-uniform certificate.
\end{remark}

\begin{corollary}[Constructive slow rank growth]\label{cor:slowrank}
Fix $W>0$, $x_{\min}>0$, and $\epsilon>0$. For any $N$ positions in
$[-W/2,W/2]$ under the exact Jakes model, direct Bessel-feature
retention admits a sufficient rank
\begin{equation}
 r=O\!\left(\frac{\log\log N}{\log\log\log N}\right),\qquad N\to\infty,
 \label{eq:slowrank}
\end{equation}
for the residual certificate at every $x\ge x_{\min}$.
In particular, a sufficient rank of order $O(\log\log N)$ exists.
The constants depend on $W,x_{\min},\epsilon$.
\end{corollary}
\begin{proof}
Put $c=\log(1+\epsilon)/4$ and $k=c e^{-c}$.
Then $t_\epsilon=1-e^{-c/r}\ge k/r$.
For $r+1\ge\sqrt{2}b$, the factorial bound $r!\ge(r/e)^r$ gives
$d_r\le4(eb/r)^{2r}$. Thus
\begin{equation}
 4r^2(eb/r)^{2r}\log(N/\eta_\epsilon)\le k^2x_{\min}
 \label{eq:slowranksufficient}
\end{equation}
implies \eqref{eq:explicitvariance}.
Set $L=\log\log(N/\eta_\epsilon)$ and
$r=\max\{1,\lceil\sqrt{2}b\rceil,\lceil L/\log L\rceil\}$
for $L>1$. This explicitly ensures $r+1\ge\sqrt{2}b$;
for sufficiently large $N$, the last term determines $r$.
The logarithm of the left side of \eqref{eq:slowranksufficient} is
$\log4+2\log r+L+2r(1+\log b-\log r)=-L+o(L)$,
which tends to minus infinity. The fixed-threshold witness extends
to all $x\ge x_{\min}$ by Theorem~\ref{prop:residual}.
\end{proof}
Corollary~\ref{cor:slowrank} gives an asymptotic construction by direct
feature retention in exact arithmetic. The numerical SVD implementation
uses the tighter computed port residuals in \eqref{eq:featurevariance}.
In implementation, one may certify both the direct-feature and SVD
constructions at the same threshold and tolerance, then select the
one using fewer latent coordinates. The experiments below use SVD.

Dense eigendecomposition costs $O(N^3)$; for $m\ll N$, the feature SVD
costs $O(Nm^2)$, residual preprocessing $O(Nm)$, and one witness
evaluation $O(N)$. A retained sample requires $O(Nr)$ arithmetic to evaluate all physical ports.

\subsection{Certified Link-Design Quantities}\label{sec:designinverse}
Let $p\in(0,1)$ be a target outage and let
$[\ell_j,u_j]$ cover $F_\R(x_j)$ simultaneously at a fixed threshold grid.
For the residual method, $\ell_j=a_r(x_j)\ell_{U,j}$ and $u_j=u_{U,j}$.
Write $\gamma_j=\theta/x_j$ in linear units and define
\begin{align}
 \gamma_{\rm nec}&=\max_{j:\ell_j>p}\gamma_j,&
 \gamma_{\rm suf}&=\min_{j:u_j\le p}\gamma_j.
 \label{eq:snrinverse}
\end{align}
If both sets are nonempty, the least SNR meeting the target obeys
\begin{equation}
 \gamma_{\rm nec}\le\gamma_{\rm req}\le\gamma_{\rm suf}.
 \label{eq:snrbracket}
\end{equation}
This follows directly from monotonicity: the first endpoint fails the
target and the second passes it on the same coverage event. An empty set leaves the corresponding endpoint unresolved.

The optional outage-capacity conversion is given in Appendix~\ref{app:capacity}.
Monotonicity converts the grid evaluations into bounds on the continuous
design quantity. Model error, sampling uncertainty, and grid spacing
all affect the interval width. The local estimate
$\Delta\gamma_{\rm dB}\simeq10\log_{10}(1+\epsilon)/d$
additionally requires an approximately constant outage slope $d$.
For any other rate threshold $\theta>0$, the same outage quantile $x_p=F_\R^{-1}(p)$ gives
$\gamma_{\rm req}(\theta)=\theta\gamma_{\rm req}(1)$:
every SNR bracket shifts by $10\log_{10}\theta$ dB.

For P2, let $[\ell_d,u_d]$ cover $F_{\R(d)}(\theta/\gamma_0)$
simultaneously for all $d=(W,N)\in\mathcal D$. Classify
\begin{equation}
 \begin{cases}
 u_d\le p &: \text{certified feasible},\\
 \ell_d>p &: \text{certified infeasible},\\
 \ell_d\le p<u_d &: \text{unresolved}.
 \end{cases}
 \label{eq:p2class}
\end{equation}
On the simultaneous coverage event, the first two classifications
are correct: feasibility concerns the receiver configuration meeting
the outage target at $\gamma_0$. An unresolved interval straddles the
target. Nondominated certified-feasible candidates form a safe
resource frontier. If every candidate is resolved, it is the true
frontier within $\mathcal D$. Otherwise, unresolved candidates must be
considered before declaring an optimum. Each geometry is evaluated with
its own covariance and certificate, so the classification also applies
to nonnested port grids.

\section{Reliable Evaluation and Numerical Certification}
\subsection{Gaussian Exponential Tilting}
Ordinary Monte Carlo resolves probabilities near $1/n$ poorly and cannot substantiate probabilities many orders below it. For a positive definite covariance $\mathbf C$ with eigenvalues $\mu_j$, choose $s\ge0$ and sample
\begin{equation}
 \h_s\sim\CN\big(\bm0,(\mathbf C^{-1}+s\I)^{-1}\big).
\end{equation}
The unbiased weight for its outage is
\begin{equation}
 w_s(\h_s)=\frac{e^{s\|\h_s\|^2}}{\det(\I+s\mathbf C)}
 \mathbf1\{\|\h_s\|_\infty^2\le x\}.
 \label{eq:IS}
\end{equation}
On the event, $\|\h_s\|^2\le Nx$, giving the deterministic bound
\begin{equation}
 0\le w_s\le b_s:=\frac{e^{sNx}}{\det(\I+s\mathbf C)}.
 \label{eq:weightbound}
\end{equation}
For a singular covariance the same formula uses its positive eigenvalues and Gaussian density ratios on its support. The determinant notation in \eqref{eq:weightbound} is still valid since zero eigenvalues contribute factors of one.

We select $s$ deterministically from the eigenvalues, before drawing samples. With $a=sx$, the implementation solves
\begin{equation}
 \sum_{j:\mu_j>0}\frac{\mu_j/x}{1+a\mu_j/x}=\frac{N}{1.6}
 \label{eq:tiltchoice}
\end{equation}
if the left side at $a=0$ exceeds the target; otherwise $a=0$. This brings typical proposed energy near the event boundary. The fixed value 1.6 approximates the optimal full-rank small-threshold constant $c_*\simeq1.5936$ derived in Appendix~\ref{app:tilt}. At finite SNR it serves as a moment-based tuning rule. All calculations use logarithmic normalizers.

A second tuning rule minimizes the analytical cap $b_s$ to reduce the range contribution to the confidence interval. With $a=sx$,
\begin{equation}
 \log b(a)=Na-\sum_j\log(1+a\mu_j/x).
\end{equation}
This function is convex. Its minimum is at zero if $\sum_j\mu_j/x\le N$; otherwise it solves
$\sum_j(\mu_j/x)/(1+a\mu_j/x)=N$.
This choice minimizes the range penalty in \eqref{eq:EB}; the variance contribution depends on the resulting weight distribution.

For numerical design we balance the variance and range terms using an
independent pilot. For a planned production count $n_\star$ and interval
failure probability $\alpha$, select a tilt from a finite candidate set
by minimizing the pilot prediction
\begin{equation}
 \widehat R(s)=\sqrt{\frac{2\widehat V_s\log(4/\alpha)}{n_\star}}
       +\frac{7b_s\log(4/\alpha)}{3(n_\star-1)},
 \label{eq:pilotradius}
\end{equation}
where $\widehat V_s$ estimates the variance of the unnormalized weight.
The production tilt is then frozen. Conditional on the pilot, independent
production samples are iid with analytical bound $b_s$. The reported
interval uses these production samples.

\subsection{Simultaneous Finite-Sample Brackets}
Let $q_k=w_k/b$ be independent samples in $[0,1]$, with mean $\bar q$ and unbiased sample variance $S_q^2$. The empirical Bernstein inequality of Maurer and Pontil \cite[Th. 4]{maurer2009}, with failure probability $\alpha/2$ allocated to each tail, gives a probability interval with radius
\begin{equation}
 d_n(\alpha)=\sqrt{\frac{2S_q^2\log(4/\alpha)}{n}}
 +\frac{7\log(4/\alpha)}{3(n-1)}.
 \label{eq:EB}
\end{equation}
The interval is $b[\bar q-d_n,\bar q+d_n]$ intersected with $[0,1]$. The analytical bound $b$ holds for every possible sample. For a prespecified family of $M$ probability estimates, using $\alpha=\alpha_{\rm all}/M$ gives simultaneous coverage at least $1-\alpha_{\rm all}$ by a union bound. Independence between different family members is unnecessary, although samples within each estimate must satisfy the stated iid assumption.

For $M$ thresholds inspected at at most $J$ declared sample counts,
allocate $\alpha=\alpha_{\rm all}/(MJ)$ to each interval. This covers
all potential looks and permits stopping according to a declared SNR
width or feasibility rule. Thresholds and tilts may be selected from
an independent pilot; conditioning on it gives the same union-bound
argument for the production family. The reported intervals use the
last completed look specified by the reporting rule.

Let $[\ell_U,u_U]$ cover the exact reduced probability $U_r$.
Theorem~\ref{prop:residual} propagates this sampling interval to
\begin{equation}
 [\ell,u]=[a_r\ell_U,u_U],\qquad \ell\le F_\R(x)\le u.
 \label{eq:propagated}
\end{equation}
When $\ell>0$, a relative numerical certificate is
\begin{equation}
 \widehat\epsilon_r=\frac{u-\ell}{\ell}.
 \label{eq:cert}
\end{equation}
It can exceed the model budget $1/a_r-1$: selecting the rank is
sampling-free, while evaluating its probability adds sampling
uncertainty. A zero lower endpoint leaves relative accuracy
unresolved but can still permit a feasibility decision if $u\le p$.
Section~\ref{sec:designinverse} converts these same intervals into
P1 and P2 decisions. The guarantee assumes exact weight evaluation;
Section~\ref{sec:protocol} describes numerical-input controls.

Refinement must preserve simultaneous coverage. Extra samples may
use the remaining predeclared looks. A changed rank or threshold
grid requires coverage by the declared family or an explicitly
allocated additional failure budget and independent production
samples. If no covered refinement remains, the decision stays
unresolved. Increasing sample count alone need not overcome a
model-error or threshold-grid limitation.

\section{Numerical Results}\label{sec:results}
\subsection{Protocol, Controls, and Reproducibility}\label{sec:protocol}
The synthetic correlated-Gaussian channel calculations use reference rate $R_0=1$ bit/s/Hz ($\theta=1$); other rates follow by effective-SNR conversion. The scalability protocol fixes $N\in\{16,32,64,128,256,512,1024\}$, $W\in\{0.5,1,2,4,8\}$, seven SNRs from 0 to 30 dB, and model budgets of 5\% and 20\%, giving 490 decisions. Positive features use $m=\max\{40,\lceil\pi W\rceil+40\}$. Their omitted diagonal bound is at most $8.73\times10^{-41}$; the largest double-precision covariance-entry reconstruction discrepancy is $3.61\times10^{-15}$. We add a $2\times10^{-12}$ residual-variance cushion. Removing this cushion changes none of the selected ranks. These sensitivity checks leave rigorous interval-arithmetic certification of the covariance and SVD inputs open.

Selected scalar witnesses are recomputed at 80 digits. Bessel identities are checked at 180 digits, and small-array residuals are checked against a complete 80-digit eigensystem. The analytical feature-tail bound and these numerical sensitivity checks are recorded separately.

Table~\ref{tab:receivers} uses one joint 95\% event for all
required-SNR intervals and their differences. Four families receive
failure budget $0.0125$ each: the 98-point broad study, the
twelve-point dense-array refinement at $10^{-5}$, its eleven-point
refinement at $10^{-3}$, and the new low-port designs. The two local
dense families have three potential looks each. Their saved moments
are reused with the smaller per-interval failure probabilities;
no samples are rerun. The reported intervals use the last completed
look, covered by the union over all prescribed looks.

The new P1 family evaluates $(W,N)=(4,10),(3.5,10),(3.25,14)$,
and $(3.125,19)$ at both outage targets. Independent pilots locate
each crossing and choose among four tilts using \eqref{eq:pilotradius}.
Each design freezes thirteen thresholds spaced 0.05 dB apart and
uses cumulative looks of $0.2$, $1$, and $2$ million samples per
threshold. Thus $\alpha=0.0125/312$ covers all potential intervals.
A design stops after a complete look when its required-SNR width
is at most 0.20 dB, or at the sample limit. All eight designs stop
at one million samples per threshold, totaling 104 million draws.

The separate aperture family uses $W\in\{0.5,1,2,4,8\}$,
$N=\lceil10W\rceil+1$, and targets $10^{-3}$ and $10^{-5}$.
Independent training freezes thirteen local thresholds at 0.1-dB
spacing per crossing. Three potential production looks receive
$\alpha=0.05/390$ each; the stopping width is 0.25 dB. P2 has its
own simultaneous budget specified below.

The model budgets serve different numerical purposes. The 20\%
rank study displays compression sensitivity; the dense P1 studies
use 5\%, and the low-port P1/P2 studies use 1\% to reduce model
uncertainty near decision boundaries. At $W=4,N=256$ and 12 dB,
these budgets select ranks 15, 16, and 17, respectively. For
$(4,10)$ all three select rank ten. They are numerical tolerances,
not reliability specifications. The reported P1 stopping rule acts
on total dB width after model and sampling propagation; a fixed
relative probability error has no universal dB equivalent because
the outage slope varies with geometry and operating point. Ranks
are certified at 12 dB for the fixed receivers and 30 dB for the
aperture curve. An accompanying reproducibility supplement will be
uploaded with the manuscript through the IEEE submission system. It
contains the implementation code, frozen protocols, seeds, reported
outcomes, validation records, and figure scripts used here.

\subsection{Large-Port Rank and Residual Scaling}
Figure~\ref{fig:scalability} and Table~\ref{tab:large} show compact certified models across the tested port counts. At 30 dB and a 20\% budget, $N=1024,W=8$ needs 26 retained modes; the 99\%-energy criterion uses 18 and the cited effective-mode count is 17. The certified, energy-based, and effective-mode counts respectively target relative outage error, retained channel energy, and a simplified diversity description.

\begin{table}[t]
\caption{20\% Model Budget: Rank Pairs at (0, 30) dB}
\label{tab:large}\centering\footnotesize
\begin{tabular}{rrrrr}\toprule
$W$&$K^*$&99\% energy&$N=64$&$N=1024$\\\midrule
0.5&3&3&(4, 6)&(4, 6)\\
1&3&4&(6, 8)&(6, 8)\\
2&5&6&(9, 10)&(9, 11)\\
4&9&10&(13, 16)&(14, 16)\\
8&17&18&(22, 25)&(23, 26)\\\bottomrule
\end{tabular}
\end{table}
\begin{figure}[t]
\centering
\subfloat[]{\includegraphics[width=\columnwidth]{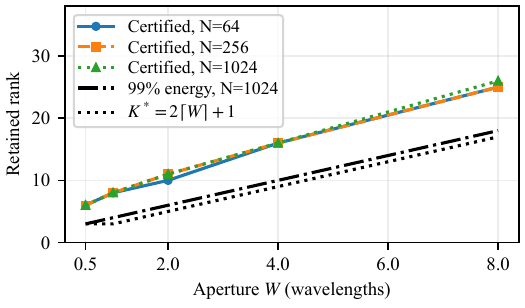}\label{fig:scalability:a}}\par
\subfloat[]{\includegraphics[width=\columnwidth]{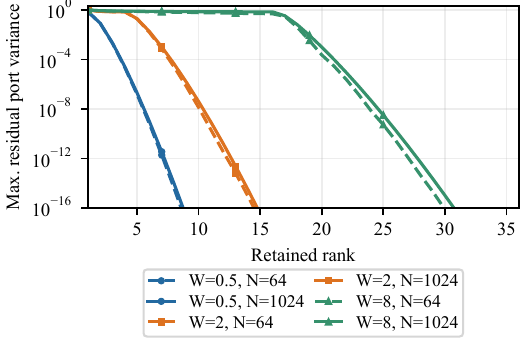}\label{fig:scalability:b}}
\caption{Large-port evaluation. (a) Sufficient ranks for a 20\% model budget at 30 dB, compared with the energy and effective-mode counts. (b) Positive residual variances before the numerical cushion. The selected ranks change modestly with port density on this finite grid. Corollary~\ref{cor:slowrank} gives a separate constructive asymptotic guarantee.}
\label{fig:scalability}
\end{figure}

At fixed aperture, adding modes rapidly reduces residual variance; changing port density has a smaller effect on this grid. The 30-dB rank-to-$K^*$ ratios at $N=1024$ are $8/3$ for $W=1$ and $26/17$ for $W=8$, illustrating dependence on aperture as well as SNR and tolerance. The measured ranks use feature SVD; the direct-feature construction of Corollary~\ref{cor:slowrank} establishes the existence of slowly growing sufficient ranks at fixed aperture.

\subsection{Required SNR and the Cost of Reducing Port Count}
Table~\ref{tab:receivers} connects the low-port P2 designs to their
required SNR through \eqref{eq:snrinverse}. For $W=4$, reducing
$N$ from 256 to ten changes the $10^{-5}$ requirement from
3.30--3.70 dB to 4.40--4.55 dB. Subtracting the intervals on their
joint confidence event gives a power cost of 0.70--1.25 dB.
The ten-port receiver at $W=3.5$ requires 4.90--5.05 dB, a
1.20--1.75-dB cost relative to the same dense four-wavelength
reference. The intervals for $(3.25,14)$ and $(3.125,19)$ complete
the SNR comparison for the small-aperture P2 candidates.

\begin{table}[t]
\caption{Required SNR and Equivalent MRC Index, $R_0=1$ bit/s/Hz}
\label{tab:receivers}\centering\footnotesize
\setlength{\tabcolsep}{3.5pt}
\begin{tabular}{rrrrr}\toprule
$W$&$N$&SNR, $p=10^{-3}$&SNR, $p=10^{-5}$&$L_{\rm eq}$ at $10^{-5}$\\\midrule
2&64&[3.50, 4.00]&[7.25, 7.75]&[4.31, 4.47]\\
4&256&[0.80, 1.00]&[3.30, 3.70]&[5.80, 6.00]\\
4&10&[1.75, 1.85]&[4.40, 4.55]&[5.42, 5.50]\\
3.5&10&[2.10, 2.20]&[4.90, 5.05]&[5.22, 5.29]\\
3.25&14&[1.90, 2.05]&[4.80, 5.00]&[5.24, 5.33]\\
3.125&19&[1.90, 2.05]&[4.85, 5.05]&[5.22, 5.31]\\\bottomrule
\end{tabular}
\par\smallskip\raggedright SNRs are in dB. All intervals share one
95\% event. The continuous index $L_{\rm eq}$ in \eqref{eq:equivmrc}
uses equal average SNR per branch or port and is rounded outward.
\end{table}

The dense-array bracket widens from its original stand-alone
3.40--3.60 dB to 3.30--3.70 dB when its failure budget is shared
with the other receivers. Figure~\ref{fig:design}(a) uses this joint
allocation. At 3.5 dB its variance and range contributions to the
EB radius are $4.23\times10^{-7}$ and $1.45\times10^{-6}$,
respectively. The range term remains the larger contribution.
At 28 GHz, the $W=4$ aperture is about 4.28 cm: its adjacent-port
spacing is 0.168 mm for 256 ports and 4.76 mm for ten ports.
These are alternative sampling densities of the same aperture.

\begin{figure}[t]
\centering
\subfloat[]{\includegraphics[width=\columnwidth]{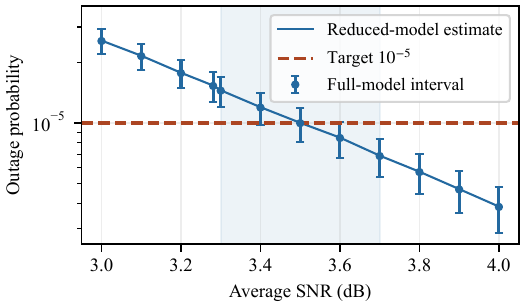}\label{fig:design:a}}\par
\subfloat[]{\includegraphics[width=\columnwidth]{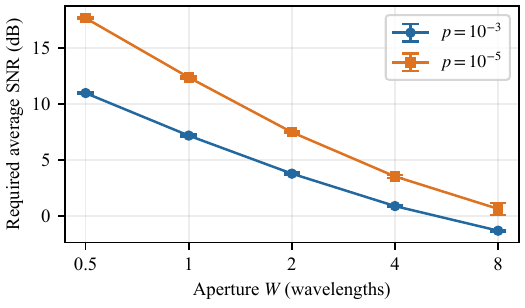}\label{fig:design:b}}
\caption{P1 dimensioning at $R_0=1$ bit/s/Hz. (a) Probability intervals for $W=4,N=256$ under the joint Table~\ref{tab:receivers} budget; shading marks the 3.30--3.70-dB required-SNR bracket. (b) A fixed-density aperture curve with $N=\lceil10W\rceil+1$, using a separate simultaneous family. For a physical-size example, $W=4$ corresponds to about 4.28 cm at 28 GHz.}
\label{fig:design}
\end{figure}

\subsection{Aperture-to-SNR Design Curve}
Figure~\ref{fig:design}(b) converts aperture into a required-SNR range at a prescribed port spacing. For $p=10^{-5}$, increasing $W$ from 0.5 to 1, 2, and 4 gives intervals [17.6,17.8], [12.3,12.5], [7.4,7.6], and [3.4,3.7] dB. At $W=8$, the interval is [0.1,1.2] dB. For $p=10^{-3}$, the corresponding intervals are [10.9,11.1], [7.1,7.3], [3.7,3.9], [0.8,1.0], and [$-1.4,-1.2$] dB. These decisions directly quantify the SNR reduction available along the specified aperture--port curve.

Eight of the ten intervals have width 0.20 dB. The $W=4$ and $W=8$ cases at $10^{-5}$ reach the fixed sample limit with widths 0.30 and 1.10 dB. The 1.10-dB width at $W=8$ reflects the remaining sampling uncertainty. Figure~\ref{fig:design}(b) fixes port density; joint aperture--port choices are examined next.

The choice $N=\lceil10W\rceil+1$ gives spacing $\lambda/10$ on
this grid. It is a dense reference curve, rather than a minimum-port
prescription. The numerical $N\approx10W$ convergence study in
\cite[Sec.~VIII-A]{gpbounds2026} provides context for this density.
We evaluate its outage accuracy separately at the present operating
points. The density is held fixed for P1, while P2 tests individual
integer port counts against a specified reliability requirement.

\subsection{Diversity References and the SNR--Aperture--Port Map}
For independent unit-mean Rayleigh power gains, SC has
$F_{\rm SC}(x)=(1-e^{-x})^L$. The MRC outage is the regularized
lower incomplete gamma function $P(L,x)$ \cite[Sec.~8.2]{dlmf}.
To compare receivers without selecting a particular integer
reference, define the continuous outage-matched index by
\begin{equation}
 P\bigl(L_{\rm eq},\theta/\gamma_{\rm req}\bigr)=p,
 \qquad L_{\rm eq}>0.
 \label{eq:equivmrc}
\end{equation}
For integer $L$, $P(L,x)=1-e^{-x}\sum_{k=0}^{L-1}x^k/k!$.
Since $P(L,x)$ decreases in $L$, the upper required-SNR endpoint
gives the lower $L_{\rm eq}$ endpoint. Table~\ref{tab:receivers}
reports this index at $10^{-5}$: reducing $W=4$ from 256 to ten
ports changes [5.80,6.00] to [5.42,5.50]. This is a
target-dependent comparison of outage distributions, not a
fractional branch count or an asymptotic diversity order.
For context, independent four-branch SC, six-branch SC, and
six-branch MRC require 12.375, 7.993, and 3.281 dB at this target.
MRC collects and coherently combines all branches, normally with
$L$ RF chains; the equal per-branch SNR normalization gives it
additional total collected energy. Ideal SC and FAS require
selection information.

For P2, Fig.~\ref{fig:portmap} summarizes screening of every integer
$N=1,\ldots,24$ at ten apertures and available SNRs 3, 4, 5, 6,
and 8 dB. This interval spans operating points where aperture
expansion is necessary within the tested set and where modest port
counts suffice. The 5-dB slice is one example within this map.
The aperture grid is denser near $W=3$ to resolve the transition
found there; $W=6$ and 8 examine larger apertures. The cap of 24
is an explicit finite search budget, not a physical upper bound.
It permits exact minima when every smaller integer is excluded;
a failed search makes no assertion about larger counts or other
placements. Ports are uniformly spaced with both endpoints included.

The 5-dB production statistics are reused with per-look failure
probability $0.025/720$. Five additional effective-SNR columns,
including the 2-bit/s/Hz case below, contain 1150 stochastic
candidates and use $\alpha=0.025/3450$ for three potential looks.
Independent pilots use $10{,}000$ samples per tilt; each production
case stops at its first resolved look or two million samples.
The new columns require 318.2 million production draws. Together
with exact single-port cases, this gives a 95\% simultaneous
statement for all 1440 tested configurations, covering
Fig.~\ref{fig:portmap} and the additional rate case below.

\begin{figure*}[t]
\centering\includegraphics[width=\textwidth]{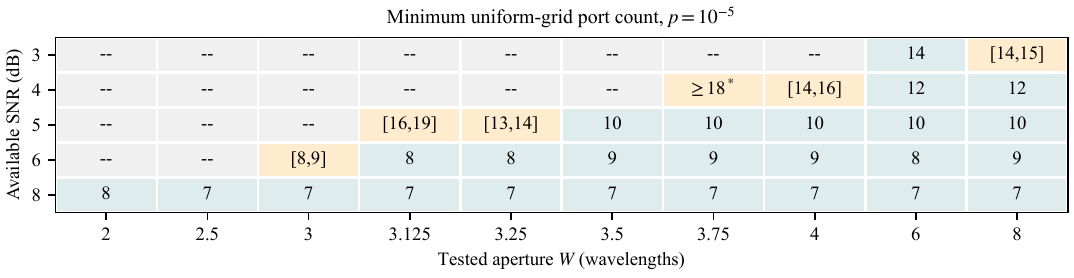}
\caption{P2 minimum-port map at $R_0=1$ bit/s/Hz and $p=10^{-5}$ on tested uniform grids. Pale cyan cells mark certified minima; pale orange cells mark unresolved minimum counts, given as integer intervals or a lower bound. Gray cells (dashes) mean all $N\le24$ fail; $\ge18^*$ means counts below 18 fail and 18--24 remain unresolved, with no certified feasible count. The map uses its own simultaneous confidence family.}
\label{fig:portmap}
\end{figure*}

The map in Fig.~\ref{fig:portmap} shows why the operating point
matters. At 3 dB, all tested counts fail for $W\le4$, while
$W=6$ needs fourteen ports and $W=8$ needs fourteen or fifteen.
At 5 dB, ten is the certified minimum at each tested aperture
from 3.5 to 8; at 8 dB, seven suffice for each tested aperture
from 2.5 to 8. Increasing $W$ or $N$ changes the port locations,
so these uniform-grid minima need not vary monotonically with
aperture. Each minimum excludes every smaller integer individually.

At 5 dB, the nondominated certified-feasible candidates are
$(3.125,19)$, $(3.25,14)$, and $(3.5,10)$. Unresolved lower
counts at the first two apertures may improve that subset frontier.
Table~\ref{tab:receivers} quantifies their corresponding power
requirements. Its separate P1 confidence allocation gives an
upper endpoint of 5.05 dB for $(3.5,10)$; the direct 5-dB P2
test resolves feasibility under the map's own allocation.

Rate also changes feasibility. At $R_0=2$ bit/s/Hz, $\theta=3$,
so 5 dB is equivalent to 0.229 dB at the reference rate. At this
operating point and $p=10^{-5}$, all $N=1,\ldots,24$ fail at each
of the ten tested apertures. For P1, changing rate instead shifts both
required-SNR endpoints by $10\log_{10}(2^{R_0}-1)$ dB: the shifts
are $-3.828$ and $4.771$ dB for 0.5 and 2 bit/s/Hz, respectively.

Independent full-covariance integration using \eqref{eq:polydisc}
confirms the nine-port failure and ten-port feasibility at $W=4$
and 8, 5 dB. Four separate controls use two million samples each
with $\alpha=0.05/4$. The ten-port intervals are
$[3.748,3.756]\times10^{-6}$ and $[3.502,3.510]\times10^{-6}$,
respectively, consistent with the primary results. Dense controls
at $W=3$, $N=121$ and 241 remain unresolved. Continuous selection
satisfies $F_{\rm cont}(x)\le F_{W,N}(x)$; a necessary continuous
aperture limit additionally requires a bound on excursions between
sampled positions. The finite-vector certificate controls the
probability calculation on the specified port set.

\subsection{Energy-Based Rank Selection and Outage Accuracy}\label{sec:energy-accuracy}
A separate motivation family uses $400{,}000$ samples per full/energy model across four cases, with $\alpha=0.05/8$. For $N=12,W=0.5$, the 99\%-energy rule selects rank three. At 5, 10, and 15 dB the full-model point estimates are $3.24\times10^{-2}$, $1.81\times10^{-3}$, and $6.46\times10^{-5}$; the energy-model/full ratios are 1.015, 1.051, and 1.169. Their simultaneous ratio intervals are approximately [0.997,1.034], [1.026,1.076], and [1.135,1.203], with the exact nonnegative model-error direction supplied by covariance order. Thus the error grows over reliability levels used in design, but the lowest-SNR numerical interval does not resolve a positive discrepancy.

For $N=64,W=2$, the six-mode energy approximation gives a ratio of 1.177 at 5 dB, where the positive-feature reference outage is $2.28\times10^{-4}$. The reference includes its controlled feature tail. KL conservatism makes this compression bias a source of potentially excess required SNR; the energy criterion alone does not quantify that margin.

\subsection{Independent-Mode and Block-Model Comparisons}\label{sec:block-model-comparison}
For the large-port design arrays we evaluate the published EDoF and WIM formulas \cite{edof2026} using
\begin{align}
 K&=\min\{N,2\lceil W\rceil+1\},\quad \beta_j=K\lambda_j/N,\\
 F_{\rm EDoF}(x)&=(1-e^{-x})^K,\quad
 F_{\rm WIM}(x)=\prod_{j=1}^K(1-e^{-x/\beta_j}).
\end{align}
\begin{figure}[t]
\centering
\subfloat[]{\includegraphics[width=\columnwidth]{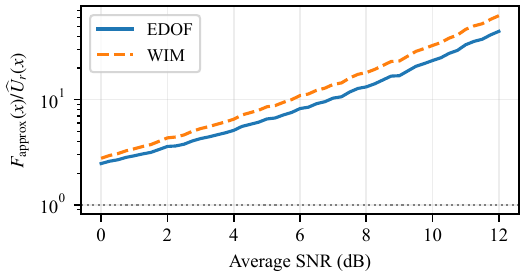}\label{fig:newbaselines:a}}\par
\subfloat[]{\includegraphics[width=\columnwidth]{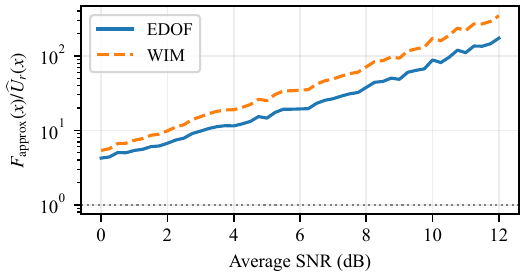}\label{fig:newbaselines:b}}
\caption{EDoF/WIM comparison on the two-array broad grid, using the v1 formulas in \cite{edof2026}. (a) $N=64,W=2$; (b) $N=256,W=4$. The denominator $\widehat U_r$ is the reduced-model point estimate.}
\label{fig:newbaselines}
\end{figure}

Figure~\ref{fig:newbaselines} compares these with the reduced-model estimates. The denominator is the reduced-model point estimate; full-model decisions use the probability brackets specified above. EDoF and WIM provide inexpensive point estimates that summarize effective spatial modes. The plotted ratios describe their agreement with the reference on the specified grid; the residual certificate separately quantifies uncertainty in that reference.

The BCM comparison implements Algorithm 1 of \cite{block2024}, eigenvalue threshold one, fixed $\mu^2=0.97$, and 0.95/0.99 sensitivity. It includes nine independent quadrature-versus-Gaussian checks. For $N=12,W=0.5$, the fit gives blocks $(8,4)$ and BCM/full ratios 0.928, 0.290, and $6.66\times10^{-6}$ at 0, 10, and 20 dB. Repeated residual eigenvalues $1-\mu^2$ differ from the small eigenvalues of the original covariance. The reported discrepancies apply to this stated calibration.

\subsection{Scattering Sensitivity and Consistency Controls}
A controlled scattering experiment uses $N=256,W=4$ and 512
independent proper-Gaussian plane waves at uniformly spaced angles,
with powers proportional to $\exp[\kappa\cos(\phi-0.7)]$.
At SNRs of 0, 10, 20, and 30 dB, respectively, the certified rank
tuples for a 20\% model budget are $(14,14,15,16)$ for $\kappa=0$,
$(13,14,15,16)$ for $\kappa=2$, and $(12,13,14,15)$ for
$\kappa=8$. Increasing the angular grid to 1024 changes covariance
entries by at most $1.23\times10^{-15}$. These are specified
finite-angle Gaussian models. A $3\times3$ isotropic square needs
rank eight at 0 and 20 dB, demonstrating the role of geometry
beyond the linear-aperture examples.

Small-array controls check the positive-feature residual against a
complete 80-digit eigensystem, the conditional integral against
independent Gaussian sampling, and the small-ball and geometry
identities in Appendices~\ref{app:polydisc} and~\ref{app:scope}. At $N=12,W=0.5$, model-error certification
and finite-sample output accuracy are separately tested. A common
20\% output-accuracy experiment gives 0.332 s for direct full-model
sampling and 0.359 s for residual screening plus sampling under a
fixed finite-look policy, excluding shared factorization and I/O.
Thus computational savings depend on the workload and estimator.
Additional controls are documented in the same reproducibility
supplement.

\section{Discussion and Research Scope}
The analysis assumes a specified centered Gaussian covariance and ideal
instantaneous selection. Extending it to measured channels requires
models for covariance uncertainty, acquisition errors, and switching
costs. A PSD uncertainty bracket offers one route to propagate covariance
estimation error. Latency and finite-blocklength decoding require
additional performance criteria beyond the outage event considered here.

Candidate ports represent alternative positions or antenna states.
They need not be separate radiators simultaneously present on an
array. The nominal Jakes covariance describes propagation correlation;
port-dependent patterns, efficiency, and coupling require an effective
channel covariance for the chosen implementation. Spacing alone does
not establish negligible coupling. The centimeter examples therefore
specify aperture size, while hardware feasibility remains dependent
on the realization.

Model compression, sampling, threshold spacing, and the number of
simultaneous decisions determine interval widths. The low-port receivers
have required-SNR widths of 0.10--0.20 dB; dense receivers retain larger
sampling costs. The wider eight-wavelength interval motivates tighter
residual-process or weight bounds. The constructive rank result explains
the availability of compact latent models as port count grows at fixed
aperture; the operating map separately quantifies physical resources.

Integer-port screening removes the large gaps of density-based
candidate sets. The resulting minima apply to the tested apertures
and uniform placement. Application-specific placement constraints
can be incorporated by screening another declared candidate set.

The certificate controls finite-threshold outage for the specified
physical-port event. A low computational rank does not determine
the exact high-SNR diversity order, and matching covariance
eigenvalues alone does not fix port-selection outage. Appendix~\ref{app:scope}
gives the corresponding order-of-limits and geometry arguments.

\section{Conclusion}
For a single-RF FAS receiver with known Gaussian covariance and ideal
port selection, the residual certificate separates the latent dimension
needed for outage evaluation from the physical resources needed to meet
a reliability target. Simultaneous probability intervals turn this
certificate into required-SNR brackets and aperture--port decisions.
At 1 bit/s/Hz and outage $10^{-5}$, reducing a four-wavelength
receiver from 256 to ten ports costs 0.70--1.25 dB; the ten-port
receiver requires 4.40--4.55 dB. A 3.5-wavelength, ten-port
receiver requires 4.90--5.05 dB and occupies about 3.75 cm at
28 GHz. These are specified points on an SNR--aperture--port map:
at 3 dB the tested small apertures fail, while at 8 dB seven
uniformly spaced ports suffice at each tested aperture from 2.5
to eight wavelengths. Rate changes are handled by the effective
SNR, and unresolved counts remain explicit. Positive features and
the constructive rank bound support compact evaluation as port
count grows. Tighter large-aperture intervals, continuous-position
bounds, and certified covariance-input arithmetic remain directions
for refinement.

\appendices
\section{Positive Definiteness of the Finite Jakes Matrix}\label{app:jakes}
For distinct normalized positions $p_i$, the identity
\begin{equation}
 J_0(2\pi(p_i-p_j))=\frac1{2\pi}\int_0^{2\pi}
 e^{\mathrm{i}2\pi(p_i-p_j)\cos\phi}\,d\phi
\end{equation}
implies
\begin{equation}
 \mathbf c^H\R\mathbf c=\frac1{2\pi}\int_0^{2\pi}
 \left|\sum_i c_i e^{-\mathrm{i}2\pi p_i\cos\phi}\right|^2d\phi.
\end{equation}
If this is zero, continuity makes the exponential polynomial vanish for every argument in $[-1,1]$. Analytic continuation and its first $N$ derivatives at zero give a Vandermonde system at distinct $p_i$, hence $\mathbf c=\bm0$. The matrix is positive definite for any fixed finite set of distinct ports, even when its smallest eigenvalues are extremely small.

\section{Small-Threshold Behavior of the Tilted Estimator}\label{app:tilt}
For fixed $N$ and $\mathbf C\succ\bm0$, set $s=c/x$, $c>0$. In the rescaled outage domain $\h=\sqrt{x}\z$, the original density becomes locally constant while the proposal tends to a circular Gaussian with rescaled covariance $\I/c$. Direct integration gives
\begin{equation}
 \lim_{x\downarrow0}\frac{\E_s[w_s^2]}{F_{\mathbf C}(x)^2}
 =\left(\frac{e^c-1}{c^2}\right)^N.
 \label{eq:BRE}
\end{equation}
For one complex coordinate, the required integral is
$\int_{|z|\le1}e^{c|z|^2}d^2z=\pi(e^c-1)/c$; the product over coordinates yields \eqref{eq:BRE}. The per-coordinate factor is minimized by the positive root of $e^c(c-2)+2=0$, namely $c_*\simeq1.5936$, yielding approximately $1.544^N$. For fixed $N$ and fixed positive definite covariance, this limit establishes bounded relative error as $x\downarrow0$; the limiting factor depends exponentially on $N$. At sufficiently small $x$, the polydisc estimator can be more accurate than Gaussian tilting.

\section{Additional Probability and Design Identities}\label{app:additional}
\subsection{A Low-Dimensional Integral}\label{app:majorants}
Let $\mathbf A_r=[\sqrt{\lambda_1}\mathbf u_1,\ldots,\sqrt{\lambda_r}\mathbf u_r]$ and $\mathbf a=\mathbf A_r\z$, with $\z\sim\CN(\bm0,\I_r)$. Conditional on $\mathbf a$, the additional components of $\sqrt{\tau_r}\mathbf v$, $\mathbf v\sim\CN(\bm0,\I_N)$, are independent at the physical ports. Thus, for $\tau_r>0$,
\begin{equation}
 L_r(x)=\E_\z\prod_{i=1}^N
 G_2\!\left(\frac{2x}{\tau_r};\frac{2|a_i|^2}{\tau_r}\right),\label{eq:conditional}
\end{equation}
where $G_2(t;\nu)$ is a noncentral chi-square CDF with two degrees of freedom and noncentrality $\nu$. This and \eqref{eq:upper_integral} integrate over $r$ complex variables, or $2r$ real variables. Eliminating the phases of the latent variables without an additional invariance argument would change the event. At $r=0$, \eqref{eq:conditional} reduces to $(1-e^{-x/\lambda_1})^N$ when $\lambda_1>0$, whereas $U_0=1$.

More generally, any diagonal $\mathbf D_r\succeq\mathbf E_r$ gives a valid lower outage bound using covariance $\R_r+\mathbf D_r$ and an analogous product integral with port-specific noise variances. The isotropic choice $\tau_r\I$ gives a directly computable majorant; optimizing a diagonal majorant may tighten the bound.

It is not necessary that $L_r$ increase with $r$. Valid monotone envelopes are
\begin{equation}
 L_r^*=\max_{0\le k\le r}L_k,\qquad
 U_r^*=\min_{0\le k\le r}U_k.
 \label{eq:envelope}
\end{equation}
In exact arithmetic $U_r$ already decreases with $r$. Sampling noise can obscure this ordering in estimated values.

\subsection{Outage-Relative Rank Rule}\label{app:covrank}
\begin{remark}[Covariance-based error control]\label{rem:covrank}
Whenever $L_r(x)>0$,
\begin{equation}
 0\le\frac{U_r-F_\R}{F_\R}\le
 \frac{U_r-L_r}{L_r}.
 \label{eq:relative}
\end{equation}
Consequently, for a requested tolerance $\epsilon>0$,
\begin{equation}
 r_\epsilon(x)=\min\{r:U_r(x)\le(1+\epsilon)L_r(x)\}
 \label{eq:rankrule}
\end{equation}
is a sufficient rank, with $r=N$ always feasible.
\end{remark}
\begin{proof}
Use $F_\R\ge L_r>0$ and $U_r-F_\R\le U_r-L_r$ in the ratio. At $r=N$ the gap is zero.
\end{proof}
The sufficient rank can be conservative when $\tau_r\I$ overstates the structured residual. It is evaluated at the operating threshold $x$.

A complementary lower bound exposes the scale of the residual directly. Choose $0<\delta<\sqrt{x}$. If $v_i=[\mathbf E_r]_{ii}$, a union bound yields
\begin{equation}
 \eta_r(\delta)=\min\!\left\{1,\sum_{i:v_i>0}e^{-\delta^2/v_i}\right\}.
\end{equation}
The triangle inequality then gives
\begin{equation}
 F_\R(x)\ge
 \left[U_r\big((\sqrt{x}-\delta)_+^2\big)-\eta_r(\delta)\right]_+.
 \label{eq:displacement}
\end{equation}
The bound holds for correlated residual ports and can be combined with $L_r$ by taking their maximum. It relates residual amplitude directly to the outage radius $\sqrt{x}$; the design experiments use Theorem~\ref{prop:residual}.

\subsection{Outage-Capacity Conversion}\label{app:capacity}
For a continuous strictly increasing outage CDF, let $x_p=F_\R^{-1}(p)$.
The same endpoints give
\begin{equation}
 \log_2\!\left(1+\bar\gamma\frac{\theta}{\gamma_{\rm suf}}\right)
 \le C_p\le
 \log_2\!\left(1+\bar\gamma\frac{\theta}{\gamma_{\rm nec}}\right),
 \label{eq:capacitybracket}
\end{equation}
where $C_p=\log_2(1+\bar\gamma x_p)$ is ideal-selection outage capacity.

\subsection{Exact Polydisc Identity}\label{app:polydisc}
\begin{proposition}[Finite-threshold small-ball bracket]\label{prop:smallball}
For $\R\succ\bm0$, let $\mathbf Z$ be uniform on the product of $N$ unit complex disks. Then
\begin{equation}
 F_\R(x)=\frac{x^N}{\det\R}
 \E\exp(-x\mathbf Z^H\R^{-1}\mathbf Z),
 \label{eq:polydisc}
\end{equation}
and
\begin{equation}
 \frac{x^N}{\det\R}e^{-x\tr(\R^{-1})/2}
 \le F_\R(x)\le\min\!\left\{1,\frac{x^N}{\det\R}\right\}.
 \label{eq:smallballbound}
\end{equation}
\end{proposition}
\begin{proof}
Integrate the complex Gaussian density over $\mathcal D_x$ and substitute $\h=\sqrt{x}\mathbf z$. The real Jacobian is $x^N$, and the unit polydisc has volume $\pi^N$, giving \eqref{eq:polydisc}. Uniform disk samples have $\E[\mathbf Z\mathbf Z^H]=\I/2$. Jensen's inequality supplies the lower bound; the integrand is at most one, giving the upper bound.
\end{proof}
The familiar asymptote $A_\R(x)=x^N/\det\R$ therefore satisfies
\begin{equation}
 1\le\frac{A_\R(x)}{F_\R(x)}\le
 \exp\!\left(\frac{x}{2}\tr\R^{-1}\right).
 \label{eq:asympgap}
\end{equation}
A sufficient condition for a multiplicative error at most $1+\epsilon$ is
\begin{equation}
 x\le\frac{2\log(1+\epsilon)}{\tr\R^{-1}}.
 \label{eq:onset}
\end{equation}
The leading small-ball power is standard; the finite-threshold form is useful here because small eigenvalues can make \eqref{eq:onset} demand an extremely large SNR. Uniform polydisc sampling of \eqref{eq:polydisc} provides an independent estimator whose relative variance tends to zero as $x\downarrow0$ for fixed $\R$. It can be inefficient before this regime is reached.

\section{Rank, Diversity, and Physical-Port Geometry}\label{app:scope}
\subsection{Diversity and the Order of Limits}
For every fixed finite positive definite covariance,
\begin{equation}
 d(\R):=\lim_{x\downarrow0}\frac{\log F_\R(x)}{\log x}=N.
 \label{eq:diversity}
\end{equation}
If $\R=\mathbf A\mathbf A^H$ has genuine rank $r$ and $\mathbf A$ has full column rank, rescaling the $r$ complex latent coordinates gives
\begin{equation}
 F_\R(x)\sim C_{\mathbf A}x^r,\quad
 C_{\mathbf A}=\pi^{-r}\operatorname{Vol}_{2r}
 \{\z:\|\mathbf A\z\|_\infty\le1\}.
 \label{eq:rankasymp}
\end{equation}
The domain is bounded because $\mathbf A$ is injective. Consequently, for $\R_\delta=\R_r+\delta\I$, every $\delta>0$ gives diversity $N$, whereas $\delta=0$ gives diversity $r$. The rank and diversity-order limits therefore do not commute.

Effective-mode representations, including the model studied in \cite{edof2026}, offer compact descriptions of spatial diversity at finite SNR. For the rank-$r$ covariance considered here, \eqref{eq:rankasymp} gives the corresponding $x^r$ scaling. If the discarded eigenvalues are instead retained and the fixed covariance is positive definite, \eqref{eq:polydisc} gives $x^N$ scaling as $x\downarrow0$. These are different limiting regimes. Our finite-threshold certificate specifies where a reduced model approximates the full physical-port outage within a requested tolerance.

The finite-port Jakes matrix with distinct positions and $W>0$ is strictly positive definite; a short argument is given in Appendix~\ref{app:jakes}. Effective rank nevertheless remains useful at finite precision and finite SNR. The numerical question is how many small modes may be neglected at a requested outage accuracy. Theorem~\ref{prop:residual} addresses that question without identifying effective rank with exact algebraic rank.

\subsection{Why Physical-Port Geometry Matters}
A unitary change of basis preserves Euclidean norm but does not preserve the physical-port maximum. Thus the eigenvalues alone do not determine $F_\R(x)$ for arbitrary $\R$. For example, take eigenvalues $(2.8,0.8,0.3,0.1)$ and construct covariances using a normalized Hadamard matrix and a normalized four-point discrete Fourier matrix. Both have the same spectrum and unit diagonal. Independent Gaussian controls give distinguishable outage probabilities for these matched quantities, as recorded in the reproducibility supplement. This example illustrates why the present certificate retains the physical selection geometry as well as the eigenvalues.

The role of correlation phase can also be illustrated directly. Let $\rho=0.4$ and
\begin{equation}
 \R_+=\begin{bmatrix}1&\rho&\rho\\\rho&1&\rho\\\rho&\rho&1\end{bmatrix},\quad
 \R_-=\begin{bmatrix}1&\rho&\rho\\\rho&1&-\rho\\\rho&-\rho&1\end{bmatrix}.
 \label{eq:counter}
\end{equation}
Both are positive definite and have identical $|R_{ij}|$, whereas $\det\R_+=0.648$ and $\det\R_-=0.392$. Their small-ball constants therefore differ by \eqref{eq:polydisc}: pairwise correlation magnitudes alone do not specify the physical-port outage. The bounds in this paper use PSD covariance order, retaining the covariance information needed for the Gaussian comparison.

\begin{IEEEbiography}[{\includegraphics[width=1in,height=1.25in,clip,keepaspectratio]{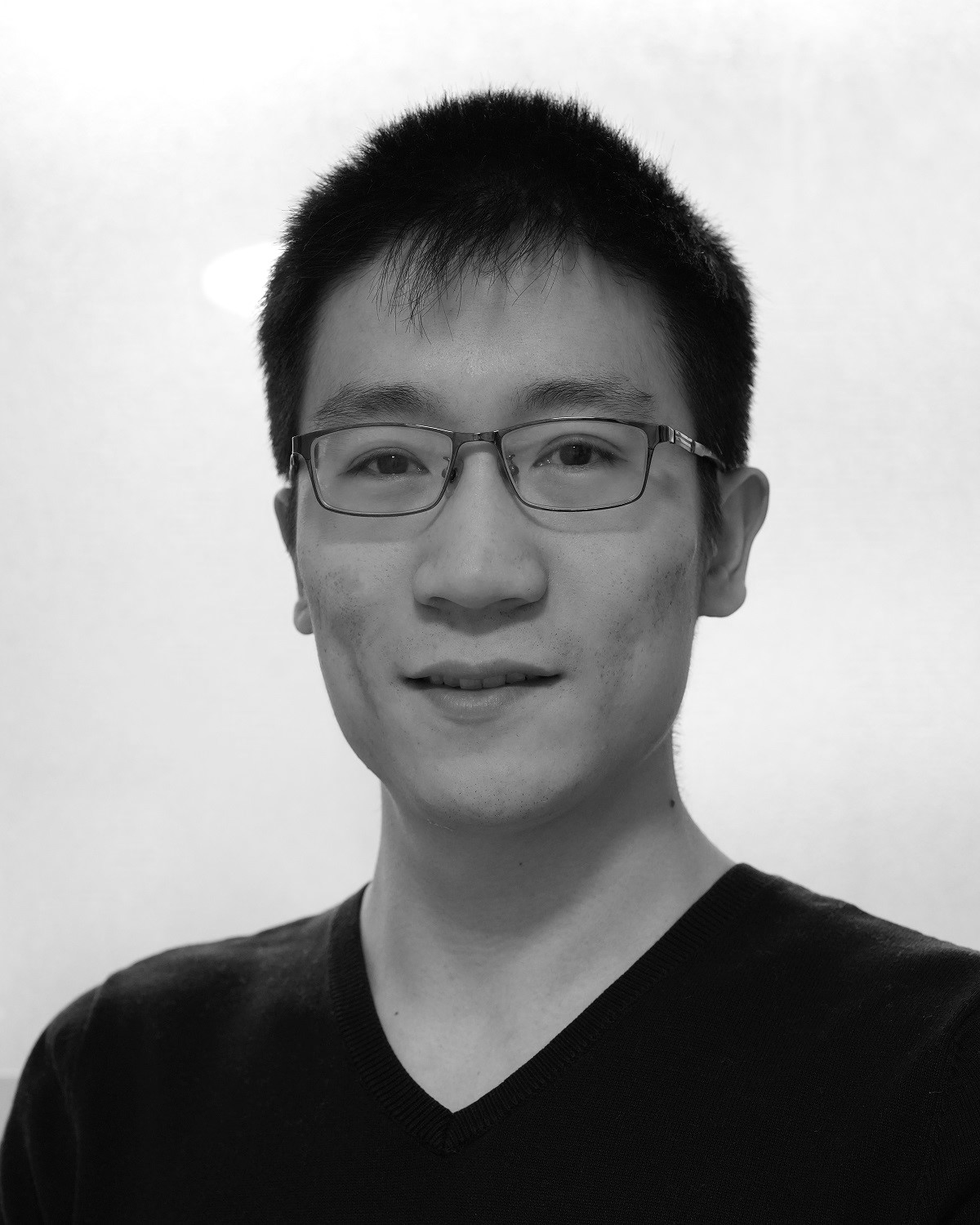}}]{Yuanjun Shen}
(Member, IEEE) was born in Xi'an, China. He received the B.Eng. degree in electronic engineering from Tianjin University, China, in 2014, the M.S. degree in electronic engineering from King's College London, U.K., in 2016, and the Ph.D. degree in electronic and electrical engineering from University College London (UCL), U.K., in 2023. From 2017 to 2023, he was a Research Assistant at UCL, London, U.K. In 2023, he joined the National Key Laboratory of Radar Detection and Sensing, Department of Electronic Engineering, Xidian University. His current research interests include fluid antenna systems, wireless communications, millimeter-wave antennas, gap waveguides, and radar systems.
\end{IEEEbiography}
\end{document}